\documentclass[10pt,twocolumn,letterpaper]{article}

\usepackage{iccv}
\usepackage{times}
\usepackage{epsfig}
\usepackage[dvipsnames]{xcolor}
\usepackage{url}            
\usepackage{booktabs}       
\usepackage{amsfonts}       
\usepackage{amssymb}        
\usepackage{nicefrac}       
\usepackage{microtype}      
\usepackage{bm}				
\usepackage{cite}			
\usepackage{graphicx}		
\usepackage{mathtools}		
\usepackage{algpseudocode}
\usepackage{algorithm}
\usepackage{amsthm}			
\usepackage{enumitem}		
\usepackage{multirow}       
\usepackage{tabularx}	    
\usepackage{hhline}
\usepackage{arydshln}		
\usepackage{enumitem}         

\newcolumntype{L}[1]{>{\raggedright\arraybackslash}p{#1}}
\newcolumntype{C}[1]{>{\centering\arraybackslash}p{#1}}
\newcolumntype{R}[1]{>{\raggedleft\arraybackslash}p{#1}}

\theoremstyle{plain} 
\newtheorem{proposition}{Proposition}

\newtheorem{theorem}{Theorem}

\newtheorem{assumption}{Assumption}

\def\defn{\,\coloneqq\,}
\def\argmin{\mathop{\mathsf{arg\,min}}} 

\def\lim{\mathop{\mathsf{lim}}} 
\def\min{\mathop{\mathsf{min}}} 
\def\max{\mathop{\mathsf{max}}}

\def\prox{\mathsf{prox}}
\def\log{\mathsf{log}}
\def\zer{\mathsf{zer}}
\def\fix{\mathsf{fix}}

\def\ebm{{\bm{e}}}

\def\sbm{{\bm{s}}}
\def\xbm{{\bm{x}}}

\def\ybm{{\bm{y}}}
\def\zbm{{\bm{z}}}

\def\zerobm{\bm{0}}
\def\Abm{{\bm{A}}}
\def\Hbm{{\bm{H}}}
\def\Dbm{{\bm{D}}}

\def\xbmast{{\bm{x}^\ast}}

\def\xbmhat{{\widehat{\bm{x}}}}
\def\Psfhat{{\widehat{\Psf}}}
\def\Gsfhat{{\widehat{\Gsf}}}
\def\Psfhat{{\widehat{\Psf}}}
\def\nablahat{{\widehat{\nabla}}}

\def\Tsf{{\mathsf{T}}}

\def\Dsf{{\mathsf{D}}}

\def\Hsf{{\mathsf{H}}}

\def\Gsf{{\mathsf{G}}}
\def\Isf{{\mathsf{I}}}
\def\Psf{{\mathsf{P}}}

\def\Usf{{\mathsf{U}}}
\def\Hsf{{\mathsf{H}}}

\def\C{\mathbb{C}}
\def\R{\mathbb{R}}
\def\E{\mathbb{E}}

\usepackage[pagebackref=true,breaklinks=true,letterpaper=true,colorlinks,bookmarks=false]{hyperref}

\iccvfinalcopy 

\def\iccvPaperID{10} 

\ificcvfinal\fi
\begin{document}

\title{Online Regularization by Denoising with Applications to Phase Retrieval}

\author{Zihui Wu \quad Yu Sun \quad Jiaming Liu \quad Ulugbek S. Kamilov\\
Washington University in St. Louis\\
{\small \tt \{ray.wu, sun.yu, jiaming.liu, kamilov\}@wustl.edu} \\
{\small \tt \textcolor{RubineRed}{https://cigroup.wustl.edu}}
}

\maketitle

\begin{abstract}

  Regularization by denoising (RED) is a powerful framework for solving imaging inverse problems. Most RED algorithms are iterative batch procedures, which limits their applicability to very large datasets. In this paper, we address this limitation by introducing a novel online RED (On-RED) algorithm, which processes a small subset of the data at a time. We establish the theoretical convergence of $\text{On-RED}$ in convex settings and empirically discuss its effectiveness in non-convex ones by illustrating its applicability to phase retrieval. Our results suggest that On-RED is an effective alternative to the traditional RED algorithms when dealing with large datasets.
\end{abstract}

\section{Introduction}
\label{Sec:Introduction}

The recovery of an unknown image $\xbm \in \R^n$ from a set of noisy measurement is crucial in many applications, including computational microscopy \cite{Tian.Waller2015}, astronomical imaging \cite{Starck.etal2002}, and phase retrieval \cite{Candes.etal2012}. The problem is usually formulated as a regularized optimization
\begin{equation}
\label{Eq:RegularizedOptimization}
\xbmhat = \argmin_{\xbm \in \R^N} \left\{f(\xbm)\right\} \quad\text{with}\quad f(\xbm) = g(\xbm) + h(\xbm),
\end{equation}
where $g$ is the data-fidelity term that ensures the consistency with the measurements, and $h$ is the regularizer that imposes the prior knowledge on the unknown image. Popular methods for solving such optimization problems include the family of proximal methods, such as proximal gradient method (PGM)~\cite{Figueiredo.Nowak2003, Daubechies.etal2004, Bect.etal2004, Beck.Teboulle2009a} and alternating direction method of multipliers (ADMM)~\cite{Eckstein.Bertsekas1992, Afonso.etal2010, Ng.etal2010, Boyd.etal2011}, due to their compatibility with non-differentiable regularizers~\cite{Rudin.etal1992, Figueiredo.Nowak2001, Elad.Aharon2006}.

\begin{figure}[t]
\centering\includegraphics[width=\linewidth]{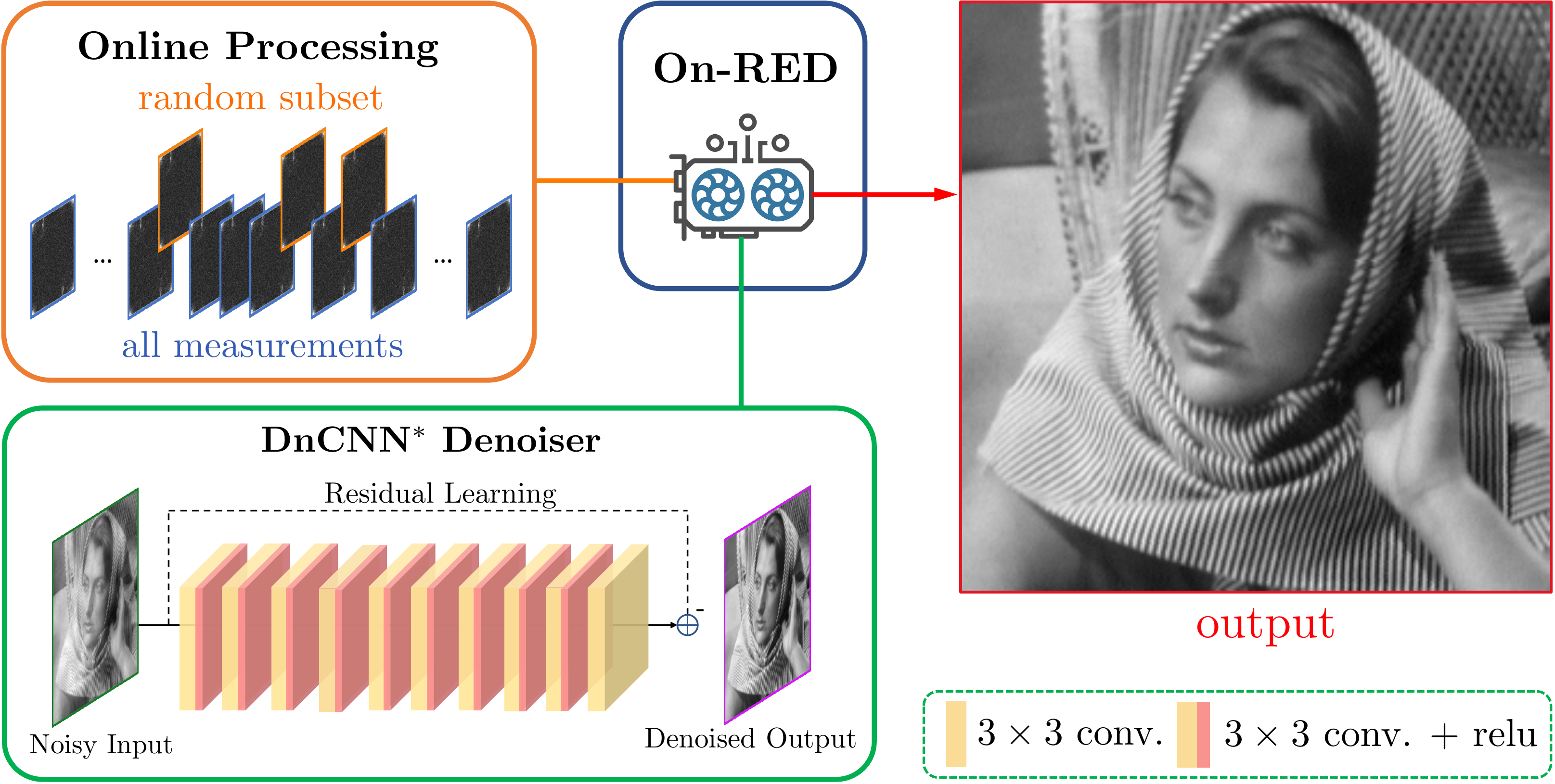}
\caption{Conceptual illustration of \emph{online regularization by denoising (On-RED)}. The proposed algorithm uses a \emph{random subset of noisy measurements} at every iteration to reconstruct a high-quality image using a \emph{convolutional neural network (CNN)} denoser.}
\label{fig:Schema}
\end{figure}

Recent work has demonstrated the benefit of using denoisers as priors for solving imaging inverse problems \cite{Sreehari.etal2016, Chan.etal2016, Brifman.etal2016, Teodoro.etal2016, Zhang.etal2017a, Meinhardt.etal2017, Kamilov.etal2017, Sun.etal2018a, Sun.etal2018b, Metzler.etal2018}. 
One popular framework, known as \emph{plug-and-play priors (PnP)} \cite{Venkatakrishnan.etal2013}, extends traditional proximal methods by replacing the proximal operator with a general denoising function. This grants PnP a remarkable flexibility in choosing image priors, but also complicates its analysis due to the lack of an explicit objective function.

An alternative strategy for leveraging denoisers is the \emph{regularization by denoising (RED)} framework~\cite{Romano.etal2017}, which formulates an explicit regularizer $h$ for certain classes of denoisers~\cite{Romano.etal2017, Reehorst.Schniter2019}. Recent work has shown the effectiveness of RED under sophisticated denoisers for many different image reconstruction tasks~\cite{Romano.etal2017,Metzler.etal2018, Reehorst.Schniter2019, Sun.etal2019a}. For example, Metzler \emph{et al.}~\cite{Metzler.etal2018} demonstrated the state-of-the-art performance of RED for phase retrieval by using the DnCNN denoiser~\cite{Zhang.etal2017}. 

Typical PnP and RED algorithms are iterative \emph{batch} procedures, which means that they processes the entire set of measurements at every iteration. This type of batch processing of data is known to be inefficient when dealing with large datasets~\cite{Bottou.Bousquet2007,Kim.etal2013}. Recently, an online variant of PnP~\cite{Sun.etal2018a} has been proposed to address this problem, yet such an algorithm is still missing for the RED framework. 

In order to address this gap, we propose an \emph{online} extension of RED, called \emph{online regularization by denoising (On-RED)}. Unlike its batch counterparts, On-RED adopts online processing of data by using only a random subset of measurements at a time (see Figure~\ref{fig:Schema} for a conceptual illustration). This empowers the proposed method to effectively scale to datasets that are too large for batch processing. Moreover, On-RED can fully leverage the flexibility offered by deep learning by using \emph{convolutional neural network (CNN)} denoisers.

The key contributions of this paper are as follows:
\begin{itemize}
\item We propose a novel On-RED algorithm for online processing of measurements. We provide the theoretical convergence analysis of the algorithm under several transparent assumptions. In particular, given a convex $g$ and nonexpansive denoiser, which does not necessarily correspond to any explicit $h$, our analysis shows that On-RED converges to a fixed point at the worst-case rate of $O(1/\sqrt{t})$.

\item We validate the effectiveness of On-RED for phase retrieval from \emph{Coded Diffraction Patterns} (CDP)~\cite{Candes.etal2012} under a CNN denoiser. Numerical results demonstrate the empirical fixed-point convergence of On-RED in this non-convex setting and show its potential for processing large datasets under nonconvex $g$.

\end{itemize}

\section{Background}
\label{Sec:Background}
In this section, we first review the problem of regularized image reconstruction and then introduce some related work.

\subsection{Inverse Problems in Imaging}

Consider the inverse problem of recovering $\xbm \in \R^n$ from measurements $\ybm \in \R^m$ specified by the linear system
\begin{equation}
\label{Eq:MeasurementSystem}
\ybm = \Hbm\xbm + \ebm,
\end{equation}
where the measurement matrix $\Hbm \in \R^{m\times n}$ characterizes the response of the system, and $\ebm$ is usually assumed to be additive white Gaussian noise (AWGN). When the inverse problem is nonlinear, the measurement operator can be generalized to a mapping $\Hbm: \R^n \rightarrow \R^m$. A common example is the problem of \emph{phase retrieval (PR)}, which corresponds the following nonlinear system
\begin{equation}
\label{Eq:PhaseRetrieval}
\ybm = \Hbm(\xbm) + \ebm, \quad\text{with}\quad \Hbm(\xbm) = |\Abm \xbm|
\end{equation} 
where $|\cdot|$ denotes an element-wise absolute value, and ${\Abm \in \C^{m \times n}}$ is the measurement matrix.

Due to the ill-posedness, inverse problems are often formulated as \eqref{Eq:RegularizedOptimization}. A widely-used data-fidelity term is the least-square loss
\begin{equation}
\label{Eq:L2norm}
g(\xbm) = \frac{1}{2}\| \ybm - \Hbm(\xbm) \|_2^2,
\end{equation} 
which penalizes the mismatch to the measurements in terms of $\ell_2$-norm. In particular, for the PR problem, the data-fidelity becomes $\frac{1}{2}\| \ybm - |\Abm \xbm| \|_2^2$, which is known to be \emph{non-convex}. Two common choices for the regularizer include the sparsity-enhancing $\ell_1$ penalty $h(\xbm) = \tau\|\xbm\|_1$ and the total variation (TV) penalty $h(\xbm) = \tau\|\Dbm\xbm\|_1$, where $\tau>0$ controls the strength of regularization and $\Dbm$ denotes the discrete gradient operator~\cite{Rudin.etal1992, Tibshirani1996, Candes.etal2006, Donoho2006, Kamilov2017}. 

Two popular methods for solving \eqref{Eq:RegularizedOptimization} are PGM and ADMM. They circumvent the differentiation of non-smooth regularizers by using a mathematical concept called \emph{proximal map}~\cite{Moreau1965}
\begin{equation}
\label{Eq:ProximalOperator}
\prox_{\tau h}(\zbm) \defn \argmin_{\xbm \in \R^n} \left\{\frac{1}{2}\|\xbm-\zbm\|_2^2 + \tau h(\xbm)\right\}.
\end{equation}
A close inspection of \eqref{Eq:ProximalOperator} reveals that the proximal map actually corresponds to an image denoiser based on regularized optimization. This mathematical equivalence led to the development of PnP and RED.

\subsection{Plug-and-play algorithms}

Consider the ADMM iteration
\begin{align}
\label{Eq:ADMM}
\zbm^k &\leftarrow \prox_{\tau g}(\xbm^{k-1} - \sbm^{k-1}) \nonumber \\
\xbm^k &\leftarrow \prox_{\tau h}(\zbm^k + \sbm^{k-1}) \\
\sbm^k &\leftarrow \sbm^{k-1} + (\zbm^k - \xbm^k) \nonumber,
\end{align}
where $k\geq1$ denotes the iteration number. In \eqref{Eq:ADMM}, the regularization is imposed by $\prox_{\tau h}: \R^n \rightarrow \R^n$, which denotes the proximal map of $h$. 

\begin{figure*}
\begin{minipage}[t]{.5\textwidth}
\begin{algorithm}[H]
\caption{$\mathsf{GM}$-$\mathsf{RED}$}\label{alg:RED}
\begin{algorithmic}[1]
\State \textbf{input: } $\xbm^0 \in \R^n$, $\tau > 0$, and $\sigma > 0$
\For{$k = 1, 2, \dots$}
\State $\nabla g(\xbm^{k-1}) \leftarrow \mathsf{fullGradient}(\xbm^{k-1})$
\State $\Gsf(\xbm^{k-1}) \leftarrow \nabla g(\xbm^{k-1}) + \tau(\xbm^{k-1}-\Dsf_\sigma(\xbm^{k-1}))$
\State $\xbm^k \leftarrow \xbm^{k-1} - \gamma \Gsf(\xbm^{k-1})$
\EndFor\label{euclidendwhile}
\end{algorithmic}
\end{algorithm}
\end{minipage}
\hspace{0.25em}
\begin{minipage}[t]{.5\textwidth}
\begin{algorithm}[H]
\caption{$\mathsf{On}$-$\mathsf{RED}$}\label{alg:OnRED}
\begin{algorithmic}[1]
\State \textbf{input: } $\xbm^0 \in \R^n$, $\tau > 0$, $\sigma > 0$, and $B \geq 1$
\For{$k = 1, 2, \dots$}
\State $\nablahat g(\xbm^{k-1}) \leftarrow \mathsf{minibatchGradient}(\xbm^{k-1}, B)$
\State $\Gsfhat(\xbm^{k-1}) \leftarrow \nablahat g(\xbm^{k-1}) + \tau(\xbm^{k-1}-\Dsf_\sigma(\xbm^{k-1}))$
\State $\xbm^k \leftarrow \xbm^{k-1} - \gamma \Gsfhat(\xbm^{k-1})$
\EndFor\label{euclidendwhile}
\end{algorithmic}
\end{algorithm}
\end{minipage}
\end{figure*}

Inspired by the equivalence that the proximal map is a denoiser, Venkatakrishnan \emph{et al.}~\cite{Venkatakrishnan.etal2013} introduced the PnP framework based on ADMM by replacing $\prox_{\tau h}$ in~\eqref{Eq:ADMM} with a general denoising function $\Dsf_\sigma : \R^n \rightarrow \R^n$
\begin{equation}
\label{Eq:PnPADMM}
\xbm^k \leftarrow \Dsf_{\sigma}(\zbm^k + \sbm^{k-1})\nonumber \\
\end{equation}
where $\sigma > 0$ controls the strength of denoising. This simple replacement enables PnP to regularize the problem by using advanced denoisers, such as BM3D~\cite{Dabov.etal2007} and DnCNN. Numerical experiments show that PnP achieves the state-of-the-art performance in many applications. Similar PnP algorithms have been developed using PGM~\cite{Kamilov.etal2017}, primal-dual splitting~\cite{Ono2017}, and approximate message passing (AMP)~\cite{Metzler.etal2016, Fletcher.etal2018}. 

Considerable effort has been made to understand the theoretical convergence of the PnP algorithms~\cite{Sreehari.etal2016, Chan.etal2016, Meinhardt.etal2017, Teodoro.etal2017, Buzzard.etal2017, Sun.etal2018a, Ryu.etal2019}. Recently, Sun \emph{et al.}~\cite{Sun.etal2018a} proposed an online PnP algorithm based on PGM, named PnP-SPGM, and analyzed its fixed-point convergence using the monotone operator theory~\cite{Bauschke.Combettes2017}. This paper extends their results to the RED framework by introducing a new algorithm and analyzing its theoretical convergence.


\subsection{Regularization by Denoising}

The RED framework, proposed by Romano \emph{et al.}~\cite{Romano.etal2017}, is an alternative way to leverage image denoisers. RED has been shown successful in many regularized reconstruction tasks, including image deblurring~\cite{Romano.etal2017}, super-resolution~\cite{Mataev.etal2019}, and phase retrieval~\cite{Metzler.etal2018}. The framework aims to find a fixed point $\xbm^\ast$ that satisfies
\begin{align}
\label{Eq:FixedPoints}
\Gsf(\xbm^\ast) = \nabla g(\xbm^\ast) + \tau (\xbm^\ast - \Dsf_\sigma(\xbm^\ast)) = 0,
\end{align}
where $\tau>0$ and $\nabla g$ denotes the gradient of $g$. Equivalently, $\xbm^\ast$ lies in the zero set of $\Gsf: \R^n \rightarrow \R^n$
\begin{align}
\label{Eq:ZeroSet}
\xbm^\ast \in \zer(\Gsf) \defn \{\xbm \in \R^n\;|\; \Gsf(\xbm) = 0\}.
\end{align}
Romano \emph{et al.} discussed several RED algorithms for finding such $\xbm^\ast$. One popular algorithm is the gradient descent (summarized in  Algorithm~\ref{alg:RED})
\begin{align}
\label{Eq:REDupdate}
\xbm^k &\leftarrow \xbm^{k-1} - \gamma (\nabla g(\xbm^{k-1}) + \Hsf(\xbm^{k-1})) \nonumber \\
&\text{with}\quad \Hsf(\xbm) \defn \tau(\xbm - \Dsf_\sigma(\xbm)),
\end{align}
where $\gamma>0$ is the step-size. They have justified $\Hsf(\cdot)$ as a gradient of some explicit function under some conditions. In particular, when denoiser $\Dsf_\sigma$ is locally homogeneous and has a symmetric Jacobian~\cite{Romano.etal2017, Reehorst.Schniter2019}, $\Hsf$ corresponds to the gradient of the following regularizer
\begin{equation}
\label{Eq:Regularizer}
h(\xbm) = \frac{\tau}{2}\xbm^\Tsf(\xbm-\Dsf_\sigma(\xbm)).
\end{equation}
By having a closed-form objective function, one can use the classical optimization theory to analyze the convergence of RED algorithms~\cite{Romano.etal2017}. On the other hand, fixed-point convergence has also been established without having an explicit objective function~\cite{Reehorst.Schniter2019, Sun.etal2019a}. Reehorst \emph{et al.}~\cite{Reehorst.Schniter2019} have shown that RED proximal gradient methods (RED-PG) converge to a fixed point by utilizing the montone operator theory. Sun \emph{et al.}~\cite{Sun.etal2019a} have established the worst-case convergence for the block coordinate variant of RED algorithm (BC-RED) under a nonexpansive $\Dsf_\sigma$. In this paper, we extend the analysis of BC-RED in~\cite{Sun.etal2019a} to the randomized processing of measurements instead of image blocks, which opens up applications requiring the processing of a large number of measurements.


\section{Online Regularization by Denoising}

We now introduce the proposed online RED (On-RED), which processes the measurements in an online fashion. The online processing of measurements is especially beneficial for problems with the following data-fidelity
\begin{equation}
\label{Eq:ComponentData}
g(\xbm) = \E[g_i(\xbm)] = \frac{1}{I}\sum_{i = 1}^I g_i(\xbm),
\end{equation}
which is composed of $I$ component functions $g_i(\xbm)$, each evaluated only on the subset $\ybm_i$ of the measurements $\ybm$. The computation of the gradient 
\begin{equation}
\label{Eq:ComponentGradient}
\nabla g(\xbm) = \E[\nabla g_i(\xbm)] = \frac{1}{I}\sum_{i = 1}^I \nabla g_i(\xbm),
\end{equation}
is proportional to the total number $I$. Note that the expectation in \eqref{Eq:ComponentData} and \eqref{Eq:ComponentGradient} is taken over a uniformly distributed random variable ${i \in \{1,\dots, I\}}$. Large $I$ effectively precludes the usage of batch GM-RED algorithms because of large memory requirements or impractical computation times. The key idea of On-RED is to approximate the gradient at every iteration by  averaging $B \ll I$ component gradients
\begin{equation}
\label{Eq:StochGrad}
\nablahat g(\xbm) = \frac{1}{B}\sum_{b = 1}^B \nabla g_{i_b}(\xbm),
\end{equation}
where $i_1, \dots, i_B$ are independent random indices that are distributed uniformly over $\{1, \dots, I\}$. The \emph{minibatch} size parameter $B \geq 1$ controls the number of gradient components used at every iteration.

Algorithm~\ref{alg:OnRED} summarizes the algorithmic details of On-RED, where the operation $\mathsf{minibatchGradient}$ computes the averaged gradients with respect to the selected minibatch components. Note that at each iteration, the minibatch is randomly sampled from the entire set of measurements. In the next section, we will present the theoretical convergence analysis of On-RED.

\section{Convergence Analysis under Convexity}
\label{Sec:Theory}

A fixed-point convergence of averaged operators is well known as Krasnosel'skii-Mann theorem~\cite{Bauschke.Combettes2017}, which was applied to the aforementioned analysis of  PnP~\cite{Sun.etal2018a} and RED algorithms~\cite{Reehorst.Schniter2019, Sun.etal2019a}. Here, our analysis extends these results to the online processing of measurements and provides explicit worst-case convergence rates for On-RED. Note that our analysis does not assume that $\mathsf{H}$ corresponds to any explicit regularizer $h$. We first introduce the assumptions necessary for our analysis and then present the main results.

\begin{assumption}
\label{As:DataFitConvexity}
We make the following assumptions on the data-fidelity term $g$:
\setlist{nolistsep,leftmargin=*}
\begin{enumerate}[label=(\alph*)]
\item The component functions $g_i$ are all convex and differentiable with the same Lipschitz constant $L > 0$.
\item At every iteration, the gradient estimate is unbiased and has a bounded variance:
$$\E[\nablahat g(\xbm)] = \nabla g(\xbm),\;\; \E[\|\nabla g(\xbm)-\nablahat g(\xbm)\|_2^2] \leq \frac{\nu^2}{B},$$
for some constant $\nu > 0$.
\end{enumerate}
\end{assumption}

\noindent
Assumption~\ref{As:DataFitConvexity}(a) implies that the overall data-fidelity $g$ is also convex and has Lipschitz continuous gradient with constant $L$. Assumption~\ref{As:DataFitConvexity}(b) assumes that the minibatch gradient is an unbiased estimate of the full gradient. The bounded
variance assumption is a standard assumption used in the analysis of online and stochastic algorithms~\cite{Ghadimi.Lan2016, Bernstein.etal2018, Xu.etal2018, Sun.etal2018a}

\begin{assumption}
\label{As:NonemptySet}
Let operator $\Gsf$ have a nonempty zero set $\zer(\Gsf) \neq \varnothing$. The distance between the the farthest point in $\zer(\Gsf)$ and the sequence $\{\xbm^{k}\}_{k=0,1,\cdots}$ generated by On-RED is bounded by constant $R_0$
$$\max_{\xbmast \in \zer(\Gsf)} \|\xbm^k-\xbmast\|_2 \leq R_0,\quad k\geq 0$$
\end{assumption}
\noindent

This assumption indicates that the iterates of On-RED lie within a Euclidean ball of a bounded radius from $\zer(\Gsf)$. 

\begin{assumption}
\label{As:NonexpansiveDen}
Given $\sigma>0$, the denoiser $\Dsf_\sigma$ is a nonexpansive operator such that
$$\| \Dsf_\sigma(\xbm) -  \Dsf_\sigma(\ybm)\|_2 \leq \|\xbm-\ybm\|_2\quad \xbm,\ybm \in \R^n,$$
\end{assumption}

\noindent
Since the proximal operator is nonexpansive~\cite{Parikh.Boyd2014}, it automatically satisfies this assumption. Nonexpansive CNN denoisers can also be trained by using spectral normalization techniques~\cite{Sun.etal2019a}. Under the above assumptions, we now establish the convergence theorem for On-RED.

\begin{theorem}
\label{Thm:ConvThm1}
Run On-RED for $t \geq 1$ iterations under Assumptions~\ref{As:DataFitConvexity}-\ref{As:NonexpansiveDen} using a fixed step-size $\gamma \in (0,1/(L+2\tau)]$ and a fixed minibatch size $B\geq1$. Then, we have
\begin{align}
\E &\left[\min_{k \in \{1, \dots, t\}} \|\Gsf(\xbm^{k-1})\|_2^2\right] \nonumber \\
&\leq\E\left[\frac{1}{t}\sum_{k = 1}^t \|\Gsf(\xbm^{k-1})\|_2^2\right]  \nonumber \\
&\leq\frac{(L+2\tau)}{\gamma} \left[\frac{\nu^2\gamma^2}{B} + \frac{2\gamma\nu}{\sqrt{B}}R_0 + \frac{R^2_0}{t} \right]. \nonumber
\end{align}
\end{theorem}
\begin{proof}
See Section~\ref{Sec:Proof}.
\end{proof}
\noindent
When $t$ goes to infinity, this theorem shows that the accuracy of the expected convergence of On-RED to an element of $\zer(\Gsf)$ improves with smaller $\gamma$ and larger $B$. For example, we can have the convergence rate of $O(1/\sqrt{t})$ by setting $\gamma = 1/(L+2\tau)$ and $B = t$
$$\E\left[\frac{1}{t}\sum_{k = 1}^t \|\Gsf(\xbm^{k-1})\|_2^2\right]\leq \frac{C}{\sqrt{t}},$$
where $C>0$ is a constant and we use the bound $\frac{1}{t}\leq \frac{1}{\sqrt{t}}$ that is valid for $t \geq 1$.


\begin{figure*}[t]
	\centering\includegraphics[width=\linewidth]{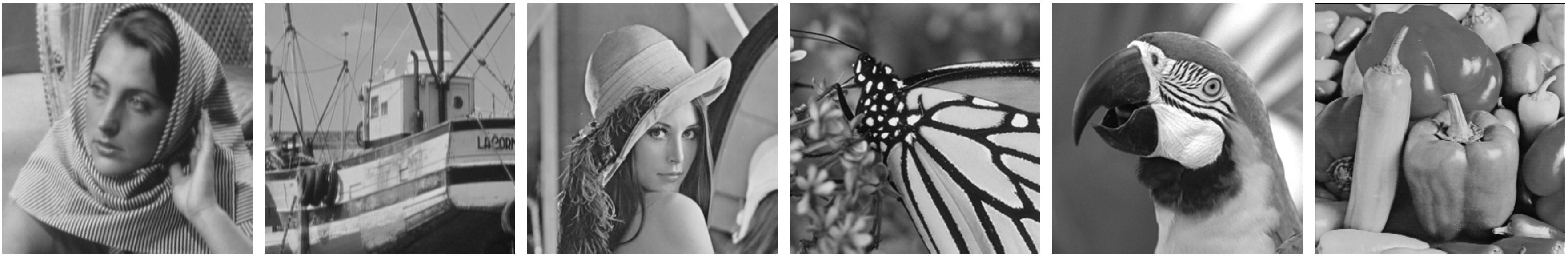}
	\caption{Test images used in the experiments. From left to right: \textit{Barbara}, \textit{Boat}, \textit{Lenna}, \textit{Monarch}, \textit{Parrot}, \textit{Pepper}.}
	\label{fig:truth}
\end{figure*}

\begin{figure*}[t]
	\centering\includegraphics[width=\linewidth]{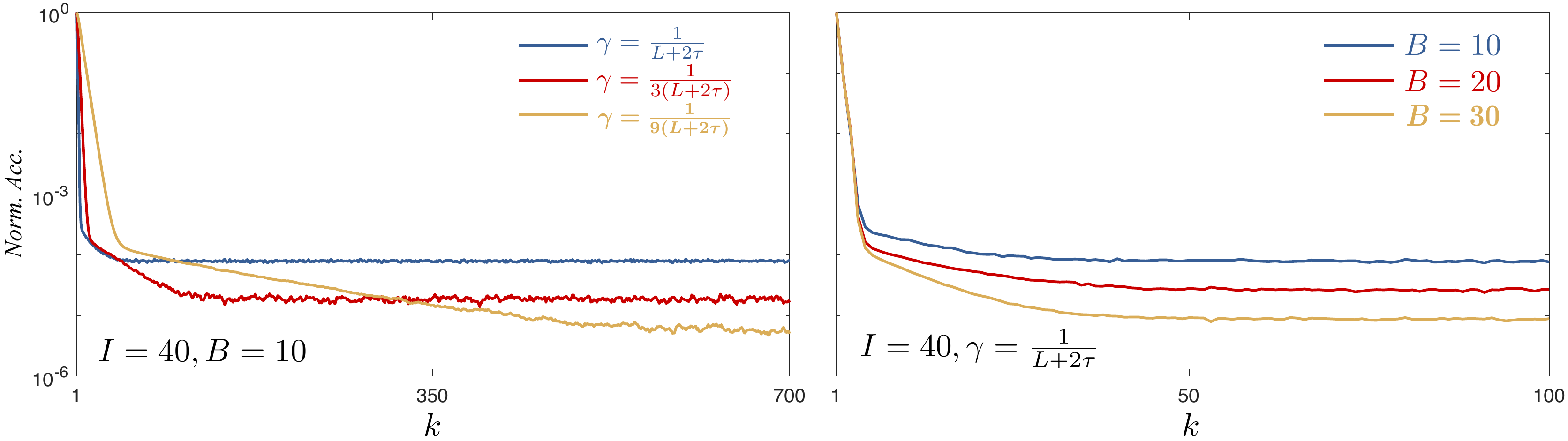}
	\caption{Illustration of the influence of $\gamma$ and $B$ on the convergence of On-RED for phase retrieval under DnCNN$^\ast$. The left plot shows the convergence results of On-RED for three different step sizes with a fixed minibatch size $B = 10$ and the right plot shows the results of On-RED for three different minibatch sizes with a fixed step size $\gamma = \frac{1}{L+2\tau}$. Both experiments draw random samples from a total of $I=40$ measurements. The plots validate that smaller $\gamma$ and larger $B$ improve the convergence accuracy in this nonconvex setting.}
	\label{fig:step_and_batch}
\end{figure*}

\section{Numerical Simulation for Phase Retrieval}

In this section, we test the performance of On-RED on a nonconvex phase retrieval problem from \emph{coded diffraction patterns (CDP)}. The state-of-the-art performance of RED for this problem was shown by Metzler \emph{et al.} \cite{Metzler.etal2018}. Here, we investigate the convergence of On-RED and show its effectiveness for reducing the per-iteration complexity of the traditional batch GM-RED. Our results show the potential of On-RED to scale to a large member of measurements under powerful denoisers that do not correspond to explicit regularizers.

\subsection{Experiment Setup}

In CDP, the object $\xbm \in \mathbb{R}^{n}$ is illuminated by a coherent light source. A random known phase mask modulates the light and the modulation code is denoted as $\bm{M}_i$ for the $i$th measurement. In this work, each entry of $\bm{M}_i$ is drawn uniformly from the unit circle in the complex plane. The light goes through the far-field Fraunhofer diffraction and a camera measures its intensity $\ybm_i \in \mathbb{R}_{+}$. Since Fraunhofer diffraction can be modeled by 2D Fourier Transform, the $i$th data-fidelity term of this phase reconstruction problem can be formulated as follows:
$$g_i(\xbm)=\frac{1}{2}\|\ybm_i-|\bm{F}\bm{M}_i \xbm|\|_{2}^{2} $$
where $\bm{F}$ denotes 2D discrete Fast Fourier Transform (FFT). The total data-fidelity term for all the measurements then becomes 
$$g(\xbm) = \E[g_i(\xbm)] = \frac{1}{I}\sum_{i = 1}^I g_i(\xbm).$$
Noticeably, this problem is well suited for On-RED because it has the same formulation as (\ref{Eq:ComponentData}).

In the experiments, we reconstruct six $256 \times 256$ standard grayscale natural images, displayed in Figure \ref{fig:truth}. The simulated measurements are corrupted by AWGN corresponding to 25 dB of input signal-to-noise ratio (SNR), defined as follows
$$\operatorname{SNR}(\hat{\ybm},\ybm)=20\log_{10}\frac{\|\ybm\|}{\|\ybm-\hat{\ybm}\|}$$
where $\hat{\ybm}$ represents the noisy vector and $\ybm$ denotes the ground truth. We also use SNR as a quantitative measure for the quality of reconstructions.

We used DnCNN$^\ast$ as our CNN denoiser for the experiments. The architecture of DnCNN$^\ast$ is illustrated in Figure \ref{fig:Schema} and was adopted from the popular DnCNN. We generated training examples by adding AWGN to images from BSD400 and applying standard data augmentation strategy including flipping, rotating, and rescaling. We used the residual learning technique where DnCNN$^\ast$ predicts the noise image from the input. The network was trained to minimize the following loss
\begin{equation}
\mathcal{L}_\theta= \frac{1}{n} \sum_{i=1}^{n} \left\{\|f_\theta(\xbm_i) - \ybm_i\|_2^2 + \|f_\theta(\xbm_i) - \ybm_i\|_1\right\},
\end{equation}
where $\xbm_i$ is the noisy input, $\ybm_i$ is the noise, and $f_\theta$ represents DnCNN$^\ast$.

The hyperparameters for experiments in \ref{Subsec:conv_exp} and \ref{Subsec:perf_exp} are listed in Table \ref{Tab:Param}. All algorithms start from $\xbm^{0}=\zerobm$, where $\zerobm \in \mathbb{R}^{n}$ is all zeros. The value of $\tau$ for each image was optimized for the best SNR performance with respect to ground truth test images. In this paper, the values of $B$ and $I$ are set only to show the potential of On-RED dealing with large datasets. 
\begin{table}[]
	\centering
	\scriptsize
	\textbf{\caption{\label{Tab:Param} List of algorithmic hyperparameters}}
	\begin{tabular*}{\linewidth}{C{1pt}C{90pt}C{55pt}C{45pt}} 			
		\toprule
		\multicolumn{2}{c}{\textbf{Hyperparameters}} & \textbf{\ref{Subsec:conv_exp}} & \textbf{\ref{Subsec:perf_exp}} \\	 
		\cmidrule(lr){1-2} \cmidrule(lr){3-4}
		$\xbm^{0}$   & initial point of reconstructions & $\zerobm$  & $\zerobm$ \\
		$\sigma$ & input noise level for denoisers & $5$ & $5$ \\
		$\tau$ & level of regularization in RED & $0.2$ & optimized \\ 
		$\gamma$   & step size & $\frac{1}{L+2\tau} \cdot \{1, \frac{1}{3}, \frac{1}{9}\}$ & $\frac{1}{L+2\tau}$ \\
		$B$ & minibatch size at every iteration & $\{10, 20, 30\}$ & $1$ \\
		$I$ & batch size & $40$ & $6$\\ 
		\bottomrule
	\end{tabular*}
\end{table}

\begin{table}[t]
	\centering
	\scriptsize
	\textbf{\caption{\label{Tab:Distance} Convergence accuracy averaged over the test images}}
	\begin{tabular*}{\linewidth}{L{20pt}C{25pt}C{25pt}C{25pt}C{21.2pt}C{21.2pt}C{21.2pt}} 			
		\toprule
		\textbf{Denoiser} & \multicolumn{3}{c}{\textbf{Step-size} ($\gamma$)} & \multicolumn{3}{c}{\textbf{Mini-batch size} ($B$)} \\ 
		\cmidrule(l){2-4} \cmidrule(lr){5-7}
		& $\frac{1}{L+2\tau}$ & $\frac{1}{3(L+2\tau)}$ & $\frac{1}{9(L+2\tau)}$ & 10 & 20 & 30 \\ 
		\cmidrule(lr){1-7}
		\textbf{TV}   & 8.65e-5 & 2.36e-5 & 9.43e-6 & 8.65e-5 & 2.81e-5 & 9.81e-6 \\
		\textbf{BM3D} & 8.01e-5 & 1.59e-5 & 9.10e-6 & 8.01e-6 & 2.72e-5 & 8.93e-6 \\
		\textbf{DnCNN$^\ast$} & 7.63e-5 & 1.94e-6 & 5.03e-6 & 7.63e-5 & 2.72e-5 & 8.88e-6 \\ 
		\bottomrule
	\end{tabular*}
\end{table}

\begin{figure*}[t]
	\centering\includegraphics[width=\linewidth]{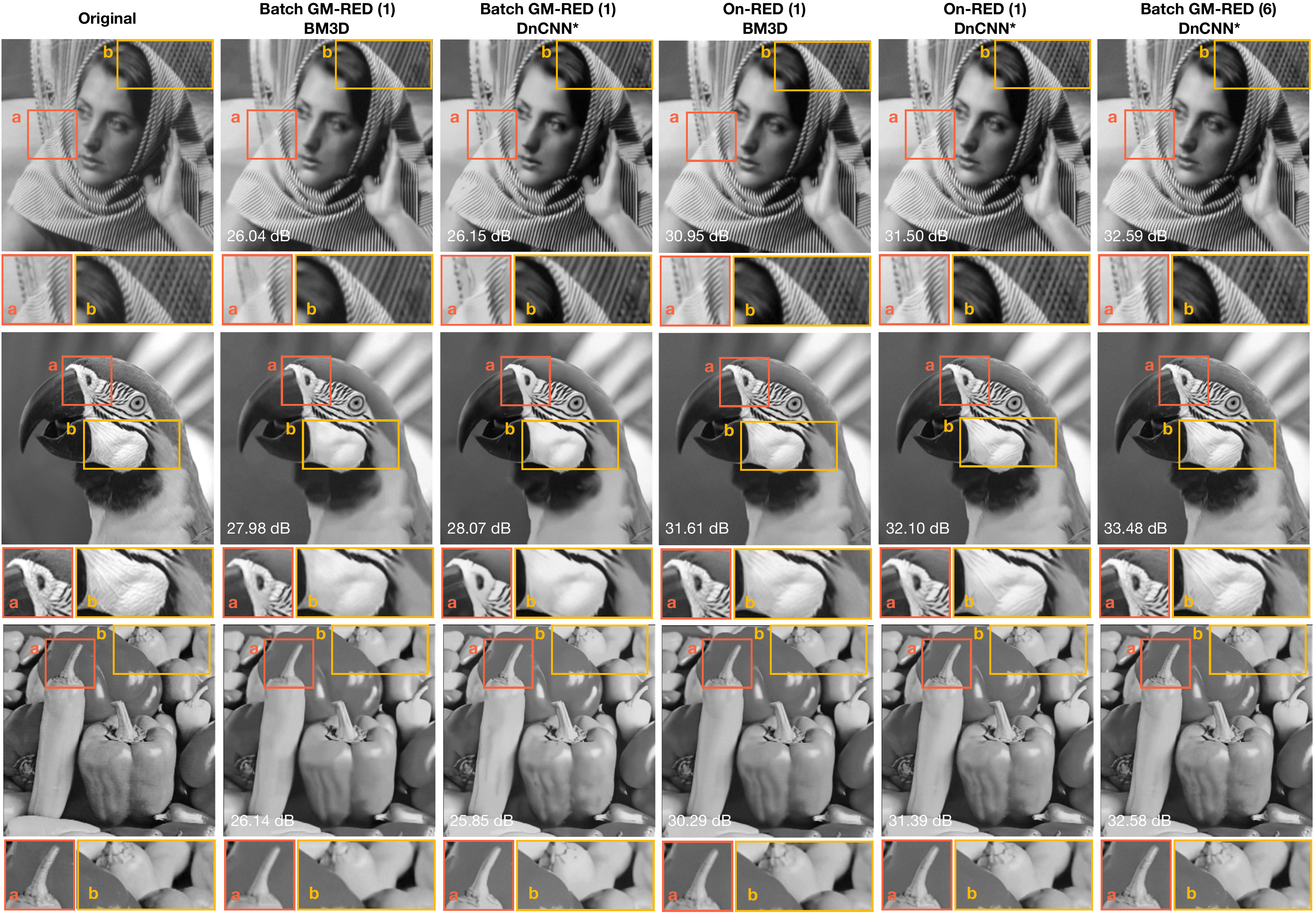}
	\caption{Visual examples of recontructed \textit{Barbara}, \textit{Parrot}, and \textit{Pepper} images by GM-RED (1), On-RED (1), and GM-RED (6) with BM3D and DnCNN$^\ast$ denoisers. The original images are displayed in the first column. The second and the third columns show the results of batch GM-RED using 1 fixed measurement. The fourth and the fifth columns present the results of On-RED using a single randomly selected measurement per-iteration out of 6 total measurements. The results of the batch algorithm using all 6 measurements are given in the last column. Differences are zoomed in using boxes inside the images. Each reconstruction is labeled by its SNR (dB) with respect to the original image. Note that On-RED (1) recovers the details lost by GM-RED (1) by approaching the performance of GM-RED (6)}
	\label{fig:examples}
\end{figure*}

\subsection{Convergence of On-RED}
\label{Subsec:conv_exp}
Theorem \ref{Thm:ConvThm1} implies that the expected accuracy improves for a smaller step size $\gamma$ and larger minibatch size $B$. In order to numerically evaluate the convergence, we define and consider the following normalized accuracy
$$\textit{Norm. Acc.} \defn \|\Gsf(\xbm^k)\|_2^2 / \|\Gsf(\xbm^0)\|_2^2$$
where $\Gsf$ is defined in (\ref{Eq:FixedPoints}). As the sequence $\{\xbm^k\}_\text{k=0,1,...}$ converges to a fixed point in $\zer(\Gsf)$,  the normalized accuracy decreases to zero. 

Figure \ref{fig:step_and_batch} (left) shows the evolution of the convergence accuracy for $\gamma \in \{\frac{1}{L+2\tau}, \frac{1}{3(L+2\tau)}, \frac{1}{9(L+2\tau)}\}$ with DnCNN$^\ast$. Here, $L$ denotes the Lipschitz constant defined in Assumption \ref{As:DataFitConvexity} and $\tau$ represents the parameter of RED. We observe that the empirical performance of On-RED using DnCNN$^\ast$ is consistent with Theorem \ref{Thm:ConvThm1}, as the accuracy improves with smaller step size. Moreover, Figure \ref{fig:step_and_batch} (right) numerically evaluates the convergence accuracy of On-RED for minibatch size $B \in \{10, 20, 30\}$. This plot shows that the convergence accuracy improves when minibatch size $B$ becomes larger. Therefore, the change of convergence accuracy with both step size $\gamma$ and minibatch size $B$ follows the same trend in Theorem \ref{Thm:ConvThm1} for this nonconvex problem.

We note that the similar trend generalizes to BM3D and TV denoisers as well. The summary in Table \ref{Tab:Distance} gives the convergence results of all three denoisers.

\subsection{Benefits of On-RED with a CNN Denoiser}
\label{Subsec:perf_exp}
In this subsection, we show the performance and efficiency of On-RED in solving CDP. To understand the potential of On-RED to scale to large datasets, we consider the scenario where the number of illuminations processed at every iteration is fixed to one.

Table \ref{Tab:SNR} provides the SNR performance of different algorithms. GM-RED (fixed 1) uses 1 fixed measurement and On-RED $(B=1)$ uses 1 random measurement out of 6 total measurements at every iteration, so they have the same per-iteration computation cost. On-RED outperforms GM-RED by 4.54 dB and 4.99 dB under BM3D and DnCNN$^\ast$, respectively, by actually using all measurements. We also note that the average SNR of \emph{stochastic gradient method (SGM)} ($B=1$) is higher than that of GM-RED (fixed 1) for both denoisers. This implies that the online processing in SGM boosts the SNR more than the regularization of GM-RED. By combining online processing and advanced denoisers, On-RED largely improves the reconstruction performance, which is close to that of the batch algorithm GM-RED (6) using all 6 measurements.

\begin{table}[t]
	\centering
	\scriptsize
	\textbf{\caption{\label{Tab:SNR} Optimized SNR for each test image in dB}}
	\begin{tabular*}{\linewidth}{L{28pt}C{26pt}C{19pt}C{19pt}C{19pt}C{19pt}C{28pt}} 	
		\toprule
		\textbf{Algorithms} & SGM & \multicolumn{2}{c}{GM-RED} & \multicolumn{2}{c}{On-RED} & GM-RED\\
		\multicolumn{1}{c}{($I=6$)} & ($B=1$) & \multicolumn{2}{c}{(fixed 1)} & \multicolumn{2}{c}{($B=1$)} & (fixed 6) \\
		\cmidrule(lr){2-2} \cmidrule(l){3-4} \cmidrule(l){5-6} \cmidrule(lr){7-7}
		\textbf{Denoisers} & --- & BM3D & DnCNN$^\ast$ & BM3D & DnCNN$^\ast$ & DnCNN$^\ast$   \\
		\midrule
		\textit{Barbara} & 27.37     & 26.04     & 26.15     & 30.95       & 31.50        & 32.59     \\
		\textit{Boat}    & 27.68     & 26.90     & 27.53     & 31.65       & 32.61        & 33.17     \\
		\textit{Lenna}   & 27.65     & 26.55     & 27.58     & 31.47       & 32.54        & 33.20     \\
		\textit{Monarch} & 27.51     & 24.76     & 26.34     & 29.66       & 31.31        & 32.63     \\
		\textit{Parrot}  & 27.20     & 27.98     & 28.07     & 31.61       & 32.10        & 33.48     \\
		\textit{Pepper}  & 27.08     & 26.14     & 25.85     & 30.29       & 31.39        & 32.58     \\
		\midrule
		\textbf{Average} & 27.42     & 26.40     & 26.92     & 30.94       & 31.91        & 32.94 	   \\
		\bottomrule
	\end{tabular*}
\end{table}

Visual illustrations of \textit{Barbara}, \textit{Parrot}, and \textit{Pepper} are given in Figure \ref{fig:examples}. It is clear that the images reconstructed by On-RED (1) preserve the features lost by GM-RED (1), such as the stripes in \textit{Barbara}, the white feather in \textit{Parrot}, and the stems in \textit{Pepper}. Moreover, these features in the reconstructed images of On-RED (1) have no visual difference from the results of GM-RED (6), as illustrated by column 4, 5, and 6. This indicates that the online algorithm approaches the image quality of the batch algorithm with a lower per-iteration complexity. 

\section{Conclusion}

In this paper, we proposed an online algorithm for the Regularization by Denoising framework. We provided the theoretical convergence proof under a few transparent assumptions and a detailed analysis in a convex problem setting. We then applied On-RED to a nonconvex phase retrieval problem from coded diffraction patterns to show its convergence. The performance of On-RED with our learning denoiser DnCNN$^\ast$ demonstrated that On-RED is well compatible with powerful denoisers that do not correspond to explicit regularizers. Our results showed that On-RED has the potential to solve data-intensive problems involving a large number of measurements by reducing per-iteration computation cost.

\section{Proof of Theorem~\ref{Thm:ConvThm1}}
\label{Sec:Proof}

We consider the following two operators
$$\Psf \defn \Isf - \gamma\Gsf \quad\text{and}\quad \Psfhat \defn \Isf - \gamma\Gsfhat \nonumber$$
where $\Psfhat$ is the online variant of $\Psf$. The iterates of On-RED can be expressed as 
$$\xbm^{k} = \Psfhat(\xbm^{k-1}) = \xbm^{k-1} - \gamma\Gsfhat(\xbm^{k-1}), \;\;\text{with}\;\; \Gsfhat = \nablahat g + \Hsf.$$
Note also the following equivalence
$$\xbm^\ast \in \zer(\Gsf) \quad\Leftrightarrow\quad \xbm^\ast \in \fix(\Psf)$$

\begin{proposition}
\label{Prp:StochP}
Consider an operater $\Psf$ and its online variant $\Psfhat$. If the data-fidelity $g(\cdot)$ satisfies Assumption~\ref{As:DataFitConvexity}, then we have
$$\E[\Psfhat (\xbm)] = \Psf(\xbm),\;\; \E[\|\Psf(\xbm)-\Psfhat(\xbm)\|_2^2] \leq \frac{\gamma^2\nu^2}{B}.$$
\end{proposition}
\begin{proof}
First, we can show 
$$\E[\Gsfhat (\xbm)] = \E[\nablahat g(\xbm)] + \Hsf(\xbm) = \Gsf(\xbm)$$
and
$$\E[\|\Gsf(\xbm) - \Gsfhat (\xbm)\|_2^2] = \E[\|\nabla g(\xbm) - \nablahat g(\xbm)\|_2^2] \leq \frac{\nu^2}{B}$$
Then, we can prove the desired result
$$\E[\Psfhat (\xbm)] = \Isf - \gamma\E[\Gsfhat(\xbm)]  = \Psf(\xbm)$$
and
$$\E[\|\Psf(\xbm) - \Psfhat (\xbm)\|_2^2] = \gamma^2\;\E[\|\Gsf(\xbm) - \Gsfhat(\xbm)\|_2^2] \leq \frac{\gamma^2\nu^2}{B}$$
\end{proof}

\begin{proposition}
\label{Prp:NonexpansiveP}
Let the denoiser $\Dsf_\sigma$ be such that it satisfies Assumption~\ref{As:NonexpansiveDen} and $\nabla g$ is L-Lipschitz continuous. For any $\gamma \in (0, 1/(L+2\tau)]$, the operator $\Psf$ is nonexpansive
$$\|\Psf(\xbm) - \Psf(\ybm)\|_2 \leq \|\xbm - \ybm\|_2\quad \forall \xbm,\ybm \in \R^n$$
\end{proposition}
\begin{proof}
The proposition is a direct result of the part (c) of the proof of Theorem 1 (Section A) in the Supplementary Material of~\cite{Sun.etal2019a} by setting $\Usf=\Usf^\Tsf=\Isf$ and $\Gsf_i = \Gsf$, which corresponds to the full-gradient RED algorithm of \eqref{Eq:REDupdate}.
\end{proof}

\noindent
Now we prove Theorem~\ref{Thm:ConvThm1} in the paper. Consider a single iteration $\xbm^k = \Psfhat(\xbm^{k-1})$, then we can write for any $\xbm^\ast\in\zer(\Gsf)$ that
\begin{align}
\label{Eq:FistStochBound}
\nonumber&\|\xbm^k - \xbmast\|_2^2 = \|\Psfhat(\xbm^{k-1})-\Psf(\xbmast)\|_2^2\\
\nonumber&= \|\Psfhat(\xbm^{k-1})-\Psf(\xbm^{k-1})+\Psf(\xbm^{k-1})-\Psf(\xbmast)\|_2^2 \\
\nonumber&= \|\Psf(\xbm^{k-1})-\Psf(\xbmast)\|_2^2 + \|\Psfhat(\xbm^{k-1})-\Psf(\xbm^{k-1})\|_2^2 \\
\nonumber&\quad\quad + 2(\Psfhat(\xbm^{k-1})-\Psf(\xbm^{k-1}))^\Tsf(\Psf(\xbm^{k-1})-\Psf(\xbmast)) \\
&\leq \|\xbm^{k-1}-\xbmast\|_2^2 - \left(\frac{\gamma}{L+2\tau}\right)\|\Gsf(\xbm^{k-1})\|_2^2 \\
\nonumber&\quad\quad + \|\Psfhat(\xbm^{k-1})-\Psf(\xbm^{k-1})\|_2^2 \\
\nonumber&\quad\quad + 2\|\Psfhat(\xbm^{k-1})-\Psf(\xbm^{k-1})\|_2 \cdot \|\Psf(\xbm^{k-1})-\Psf(\xbmast)\|_2,
\end{align}
where we use the Cauchy-Schwarz inequality and adapt the bound (14) in the part (d) of the proof of Theorem 1 (Section A) in the Supplementary Material of~\cite{Sun.etal2019a} by setting $\Usf=\Usf^\Tsf=\Isf$ and $\Gsf_i = \Gsf$. According to Assumption~\ref{As:NonemptySet} and Proposition~\ref{Prp:NonexpansiveP}, we have
\begin{equation}
\label{Eq:NonexpansiveFromInitial}
\|\Psf(\xbm^{k-1})-\Psf(\xbmast)\|_2 \leq \|\xbm^{k-1}-\xbmast\|_2 \leq R_0.
\end{equation}
Additionally, by using Jensen's inequality, we can have for all $\xbm \in \R^n$ that
\begin{align}
\label{Eq:JensenSimplification}
\E&\left[\|\Psf(\xbm)-\Psfhat(\xbm)\|_2\right] = \E\left[\sqrt{\|\Psf(\xbm)-\Psfhat(\xbm)\|_2^2}\right] \nonumber \\
&\leq \sqrt{\E\left[\|\Psf(\xbm)-\Psfhat(\xbm)\|_2^2\right]} \leq \frac{\gamma \nu}{\sqrt{B}}.
\end{align}
By rearranging and taking a conditional expectation of~\eqref{Eq:FistStochBound} and using these bounds, we can obtain
\begin{align*}
\E&\left[\|\xbm^k-\xbmast\|_2^2 - \|\xbm^{k-1}-\xbmast\|_2^2 \mid \xbm^{k-1}\right] \\
\nonumber&\leq \frac{2\gamma \nu}{\sqrt{B}}R_0 + \frac{\gamma^2 \nu^2}{B} - \left(\frac{\gamma}{L+2\tau}\right)\|\Gsf(\xbm^{k-1})\|_2^2,
\end{align*}
which can be reorganized as
\begin{align*}
\|\Gsf(\xbm^{k-1})&\|_2^2 \leq \left(\frac{L+2\tau}{\gamma}\right)\Big[\frac{\gamma^2\nu^2}{B} + \frac{2\gamma\nu}{\sqrt{B}}R_0  \\
&+\E\left[\|\xbm^{k-1}-\xbmast\|_2^2 - \|\xbm^k-\xbmast\|_2^2 \mid \xbm^{k-1}\right]\Big].
\end{align*}
By averaging the inequality over $t \geq 1$ iterations, taking the total expectation, and dropping the last term, we obtain
\begin{align*}
\E&\left[\frac{1}{t}\sum_{k = 1}^t \|\Gsf(\xbm^{k-1})\|_2^2\right] \\
&\leq \frac{L+2\tau}{\gamma} \left[\frac{\gamma^2 \nu^2}{B} + \frac{2\gamma \nu }{\sqrt{B}}R_0 + \frac{R_0^2}{t}\right]
\end{align*}
where we apply the law of total expectation and Assumption~\ref{As:NonemptySet}. This establishes the Theorem~\ref{Thm:ConvThm1}. 




{\small
\bibliographystyle{ieee}
\bibliography{egbib}
}

\end{document}